\documentclass[aps,pra,a4paper, superscriptaddress,twocolumn, 10pt]{revtex4-1} 
\usepackage{bbm, amsmath, amssymb, amsthm, bm, textcomp,graphicx,color,tabularx}
\usepackage[utf8]{inputenc}
\usepackage[T1]{fontenc}
\usepackage[english]{babel}
\usepackage{color}

\theoremstyle{definition}
\newtheorem{theorem}{Observation}

\newcommand{\red}[1]{\textcolor{blue}{#1}}     
\newcommand{\mcC}{\mathcal{C}}

\begin{document}
\title{Does a large quantum Fisher information imply Bell correlations?}
\author{Florian Fr\"owis}
\affiliation{Department of Applied Physics, University of Geneva, 1211 Geneva, Switzerland}
\author{Matteo Fadel}
\affiliation{Department of Physics, University of Basel, Klingelbergstrasse 82, 4056 Basel, Switzerland}
\author{Philipp Treutlein}
\affiliation{Department of Physics, University of Basel, Klingelbergstrasse 82, 4056 Basel, Switzerland}
\author{Nicolas Gisin}
\affiliation{Department of Applied Physics, University of Geneva, 1211 Geneva, Switzerland}
\author{Nicolas Brunner}
\affiliation{Department of Applied Physics, University of Geneva, 1211 Geneva, Switzerland}

\begin{abstract}
  The quantum Fisher information (QFI) of certain
  multipartite entangled quantum states is larger than what is
  reachable by separable states, providing a metrological advantage. Are these nonclassical
  correlations strong enough to potentially
  violate a Bell inequality? Here, we present evidence from two
  examples. First, we discuss a Bell inequality designed for
  spin-squeezed states which is violated
  only by quantum states with a large QFI. Second, we relax a
  well-known lower bound on the QFI to find the Mermin Bell inequality as a
  special case. However, a fully general link between QFI and Bell
  correlations is still open.
\end{abstract}
\date{\today}

\maketitle

\section{Introduction}
\label{sec:motivation}

The quantum Fisher information
(QFI) is an important quantity in the
geometry of Hilbert spaces \cite{Braunstein_Statistical_1994,Uhlmann_gauge_1991} and has implications for the
foundations of quantum mechanics
\cite{Toth_Quantum_2014,Taddei_Quantum_2013,Frowis_Measures_2012} as well as for quantum
metrology \cite{Helstrom_Quantum_1976,Holevo_Probabilistic_2011,Pezze_2018} and
quantum computation \cite{Shimizu_Necessity_2013,Demkowicz-Dobrzanski_Quantum_2015a}. A well-studied case is $\rho_t = \exp(-i A t)
\rho_0 \exp(i A t)$, where $\rho_0$ is an $N$-partite qubit state and $A$ is a
local operator $A = \sum_{i = 1}^{N} A^{(i)}$ (with fixed operator
norm $\lVert
A^{(i)} \rVert_{\infty} = 1/2$ for convenience). Then, the QFI $\mathcal{F}(\rho,A)$ is
a function of $\rho_0 \equiv \rho$ and $A$. This is a typical situation in quantum
metrology, where $A$ is the generator of a small perturbation (like a
weak external magnetic field) whose
strength we would like to measure as precisely as possible. The
celebrated quantum Cramér-Rao bound implies that a large QFI is necessary for
a high sensitivity \cite{Helstrom_Quantum_1976,Helstrom_1967}. It is well known that certain entangled states
allows one to go beyond the so-called standard quantum
limit for separable states
\cite{Caves_Quantummechanical_1981,Huelga_Improvement_1997,Giovannetti_Quantum_2006}. More
concretely, it was shown \cite{Sorensen_Manyparticle_2001,Pezze_Entanglement_2009} that
$\mathcal{F}(\rho,A) > N$ implies entanglement between the qubits.
The larger the QFI, the larger the entanglement depth of the
state \cite{Sorensen_Entanglement_2001,Toth_Multipartite_2012,Hyllus_Fisher_2012}.

Among the nonclassical properties of quantum systems, Bell
correlations are of particular importance. On the fundamental side,
quantum states exhibiting Bell correlations potentially violate Bell
inequalities, thereby proving that nature cannot be modeled with local
(hidden) variables \cite{Brunner_Bell_2014}. This insight can be used to design device-independent
protocols for quantum applications such as secure communication \cite{Acin_DeviceIndependent_2007} or
random number generation \cite{Pironio_2010,Colbeck}. Every quantum state with Bell
correlations is entangled, but not every entangled quantum state
necessarily has Bell correlations \cite{Werner_Quantum_1989,Barrett_2002,Bowles2016}. Hence,
the latter represents a strictly stronger form of quantum correlations.

In the present work we ask whether there exists a connection between large QFI and Bell correlations. Intuitively such a connection can be motivated by the following observation. Both large QFI and Bell correlations are properties of a quantum state associated to specific measurements. That is they cannot be a property of a quantum state (or measurement) alone, but require the judicious combination of states and measurements. 

More specifically, we investigate here whether quantum states with a high enough QFI generically exhibit Bell
correlations. If this is the case, then are the same measurements that reveal Bell
correlations (potentially in a device-dependent manner) useful to show
the presence of a large QFI? While we do not provide a fully general
answer to these questions, we discuss two examples that hint at
affirmative answers. To this end, we linearize a well-known lower bound on the
QFI. First, we take its elements to start with an ansatz for a Bell
inequality, which turns out to be of the form of multipartite Bell
inequalities based on two-body correlators recently introduced by Tura \textit{et
  al.}~\cite{Tura_Detecting_2014}. Considering a multi-setting extension of this Bell inequality \cite{WagnerPRL}, we show that (i) only quantum states with $\mathcal{F}(\rho,A) > N$ (i.e. beating the standard limit of separable states) can potentially violate the inequality, and (ii) any quantum state with $\mathcal{F}(\rho,A) > 3N$ will violate the Bell inequality. Notably, the same measurements that witness the presence of Bell correlations also
demonstrate a large QFI.

The second type of linearization is a
relaxation of the QFI bound. For a special
case which is optimal for the Greenberger-Horne-Zeilinger (GHZ)
state, we show that one side of this linear bound becomes the Bell
operator for the Mermin inequality \cite{Mermin_Extreme_1990}, another multipartite Bell
inequality specially suited to detect Bell correlations of GHZ
states. Again, a very large QFI is necessary for the violation of the
Bell inequality and the same measurements that show large QFI are
sufficient choices for a potential Bell inequality violation.

\section{Bell operators from a QFI bound}
\label{sec:bell-operator-from}

The QFI $\mathcal{F}(\rho,A)$ is a nonlinear quantity that is defined by measuring the
infinitesimal change of $\rho$ evolving under $U= \exp(-iAt)$ with the Bures
distance $s_B$ in state space,
\begin{equation}
\label{eq:11}
ds_B = \frac{1}{2} \sqrt{\mathcal{F}(\rho,A)} dt.
\end{equation}
While the exact value of $\mathcal{F}(\rho,A)$ is generally only
accessible with complete knowledge about $\rho$ and $A$
\footnote{Given the spectral decomposition \unexpanded{$\rho = \sum_k
    p_k |\psi_k \rangle \langle\psi_k |$}, the QFI reads
    \unexpanded{$\mathcal{F}(\rho,A) = 2 \sum_{k,l}  (p_k - p_l)^2/(p_k + p_l)
   | \langle \psi_k | A|  \psi_l \rangle |^2 $}.}, there are
powerful lower bounds based on relatively simple measurements. For
example, a tighter version of the Heisenberg uncertainty relation
holds for all hermitian operators $B$
\cite{Hotta_Quantum_2004,Pezze_Entanglement_2009,Frowis_Tighter_2015} 
\begin{equation}
\label{eq:12}
\mathcal{F}(\rho,A) \geq \frac{\langle i [A,B] \rangle_{\rho}^2}{\langle (B - \langle B \rangle_{\rho})^2 \rangle_{\rho}},
\end{equation}
where, in the following, we restrict ourselves to  $\langle B
\rangle_{\rho} = 0$ without loss of generality. There is always an
operator $B$ that makes inequality \eqref{eq:12} tight for given
$\rho,A$. Hence, a well-chosen $B$ allows to optimally bound the QFI.

Bell inequalities are bounds on local variable
models. Violations of these inequalities are possible in
quantum mechanics and imply the presence of Bell
correlations. For our purpose, it is sufficient to consider symmetric
Bell inequalities of $N$ parties. Following Ref.~\cite{Tura_Detecting_2014}, we define the symmetrized $k$-body correlators 
\begin{equation}
\label{eq:16}
\mathcal{C}_{j_1,\dots,j_k} = \sum_{\substack{i_1,\dots,i_k = 1\\ \text{all $i$ different}}}^N \left\langle M_{j_1}^{(i_1)} \dots M_{j_k}^{(i_k)}  \right\rangle ,
\end{equation}
where $M_j^{(i)}$ is the measurement operator for setting $j$ at site $i$. Suppose we have
$d$ measurement settings per party. Then, general linear,
symmetric Bell inequalities are of the form 
\begin{equation}
\label{eq:13}
\sum_{k = 1}^N \sum_{j_1, \dots, j_k = 0}^{d - 1} a_{j_1,\dots,j_k} \mathcal{C}_{j_1,\dots,j_k} + a_0\geq 0,
\end{equation}
where $a_{j_1,\dots,j_k},a_0 \in \mathbb{R}$. They are
fulfilled by any local hidden variable model. Here, we are
interested in nontrivial Bell inequalities, that is, in those that are
violated by some quantum states.

Assuming a connection
between large QFI and Bell correlations, one could directly try to
turn the right hand side of Eq.~\eqref{eq:12} into a Bell
inequality up to an additional local bound $a_0$. Every
symmetric, multipartite operator can be expressed in a basis of products of Pauli
operators, and its expectation value can be written as a function of
correlators \eqref{eq:16}. However, the nonlinear terms in
Eq.~\eqref{eq:12} render this approach difficult. Therefore, we propose the linear ansatz 
\begin{equation}
\label{eq:14}
\alpha + \beta \langle B^2 \rangle_{\rho} - \gamma \langle C
\rangle_{\rho} \geq 0,
\end{equation} with $C = i [A,B]$ and $\alpha,\beta,\gamma > 0$. The idea is that if $\langle
B^2 \rangle_{\rho}$ is sufficiently small and $\langle C\rangle_{}$
sufficiently large then inequality \eqref{eq:14} can be
violated which implies Bell correlations and a large QFI via Eq.~\eqref{eq:12}.

This approach turns out to be successful for spin-squeezed states
\cite{Kitagawa_Squeezed_1993}. For concreteness, we choose $A = S_z =
\frac{1}{2}\sum_i \sigma_z^{(i)}$ and $B = S_y = \frac{1}{2}\sum_i
\sigma_y^{(i)}$, that is, collective spin operators. From the
well-known SU(2) commutation relations, one has $C = S_x =
\frac{1}{2} \sum_i \sigma_x^{(i)}$. An $N$-partite qubit state is
called spin-squeezed if 
\begin{equation}
\label{eq:18}
\xi^2 = \frac{N \langle S_y^2 \rangle}{\langle S_x \rangle^2} < 1,
\end{equation}
potentially after a suitable change of collective coordinates. Hence, with our choices
for $A$ and $B$, Eq.~\eqref{eq:14} seems to be a promising candidate
for a Bell inequality that can be violated with spin-squeezed states. 

However, a direct translation of $S_x$ and
$S_y^2$ into measurement settings cannot lead to nontrivial Bell
inequalities because then an LHV model can minimize $\langle S_y^2 \rangle$
and maximize $\langle S_x \rangle$ independently of each other. To
couple the two we introduce new measurement bases for every party $i$
\begin{equation}
\label{eq:15}
\begin{split}
  M_0^{(i)} =&   \cos\phi\,
  \sigma_y^{(i)}+ \sin \phi\, \sigma_x^{(i)}  \\
  M_1^{(i)} =&   \cos\phi\,
  \sigma_y^{(i)}- \sin \phi\, \sigma_x^{(i)}.
\end{split}
\end{equation}
We note that $4 \sin  \phi\langle S_x \rangle = \red{ \mathcal{C}_0 -
\mathcal{C}_1}$ and $4 \cos^2\phi \langle S_y^2 \rangle = N \cos^2 \phi
+ \frac{1}{4}(\mathcal{C}_{00} + 2 \mathcal{C}_{01} +
\mathcal{C}_{11})$. Inserting these relations in Eq. \eqref{eq:14} we obtain an inequality of the class recently
studied by Tura \textit{et
  al.}~\cite{Tura_Detecting_2014}, who show that Eq.~\eqref{eq:14}
constitutes a Bell inequality if $\alpha = 2N \sin^2
\phi, \beta = 8 \cos^2 \phi$ and $\gamma = 4 \sin \phi$. It reads
\begin{equation}
\label{eq:17}
C_0 - C_1 + \frac{1}{2}C_{00} + C_{01} + \frac{1}{2} C_{11} + 2 N \geq 0.
\end{equation}

Under the restriction that the measurement settings $M_j^i$ are identical for all parties, the choice of Eq. \eqref{eq:15} turns out to be the most general parametrization. The Bell inequality can then be written as a Bell correlation witness which requires collective spin measurements only~\cite{Schmied_2016,WagnerPRL}. With the definition of the scaled second moment $\zeta^2 = \langle S_y^2 \rangle / (N/4)$, and of the scaled contrast $\mcC = \langle S_x \rangle / (N/2)$, the inequality becomes 
\begin{equation}\label{eq:1m}
\zeta^2 \geq \dfrac{1}{2} \left( 1 - \sqrt{1-\mcC^2} \right) \, .
\end{equation}
From the fact that $\xi^2=\zeta^2/\mcC^2$ we can express Eq.~\eqref{eq:1m} as a function of $\xi^2$ and $\mcC$, and observe that (see Fig.~\ref{fig:1}): i) for $\xi^2 \leq 1/4$ the inequality is always violated, independently on $\mcC$, ii) for $1/4 < \xi^2 < 1/2$ a minimal $\mcC$ is needed to violate the inequality, iii) for $\xi^2 \geq1/2$ the inequality is never violated. This implies that only states with $\mathcal{F}(\rho,S_z)>2N$ are able to
violate inequality \eqref{eq:17}. Moreover, all states with $\mathcal{F}(\rho,S_z)>4N$ will give violation, that is for sufficiently squeezed states, such as the one-axis and the two-axes twisted spin-squeezed state \cite{Kitagawa_Squeezed_1993}.

 It turns out that these results can be improved by considering a multi-setting generalization of the Bell inequality \eqref{eq:17} presented in Ref.~\cite{WagnerPRL}. Again, this inequality can be written as a Bell correlation witness which requires collective spin measurements only. Specifically, consider the family of $m$-settings inequalities
\begin{equation}\label{eq:3m}
\sum_{k=0}^{m-1} \alpha_k C_k + \dfrac{1}{2} \sum_{k,l} C_{k,l} + \beta_c \geq 0 \;,
\end{equation}
with $\alpha_k = m-2k-1$ and $\beta_c=\lfloor m^2 N/2 \rfloor$. This inequality can again be written as a witness which in the limit $m \rightarrow \infty$ takes the form \cite{WagnerPRL}
\begin{equation}\label{eq:2m}
\zeta^2 \geq 1 - \dfrac{\mcC}{\text{arctanh}(\mcC)} \;,
\end{equation}
which holds for all states featuring local correlations. Performing the same analysis as above, we observe that (see Fig.~\ref{fig:1}): i) for $\xi^2 \leq 1/3$ the inequality is always violated, independently on $\mcC$, ii) for $1/3 < \xi^2 < 1$ a minimal $\mcC$ is needed to violate the inequality. 

To conclude, from Eqs.~\eqref{eq:12}, \eqref{eq:18} and \eqref{eq:2m}, we see that $\mathcal{F}(\rho,S_z)>N$ is a necessary condition for violating the Bell inequality \eqref{eq:3m}. That is, only quantum states that beat the standard limit of separable states can potentially violate the Bell inequality. Moreover, the condition $\mathcal{F}(\rho,S_z)>3N$ is sufficient for violating the Bell inequality, i.e. all states satisfying it will give violation. Again, here the measurement settings are assumed to be identical for all parties.

\begin{figure}
	\centering
	\includegraphics[width=0.8\columnwidth]{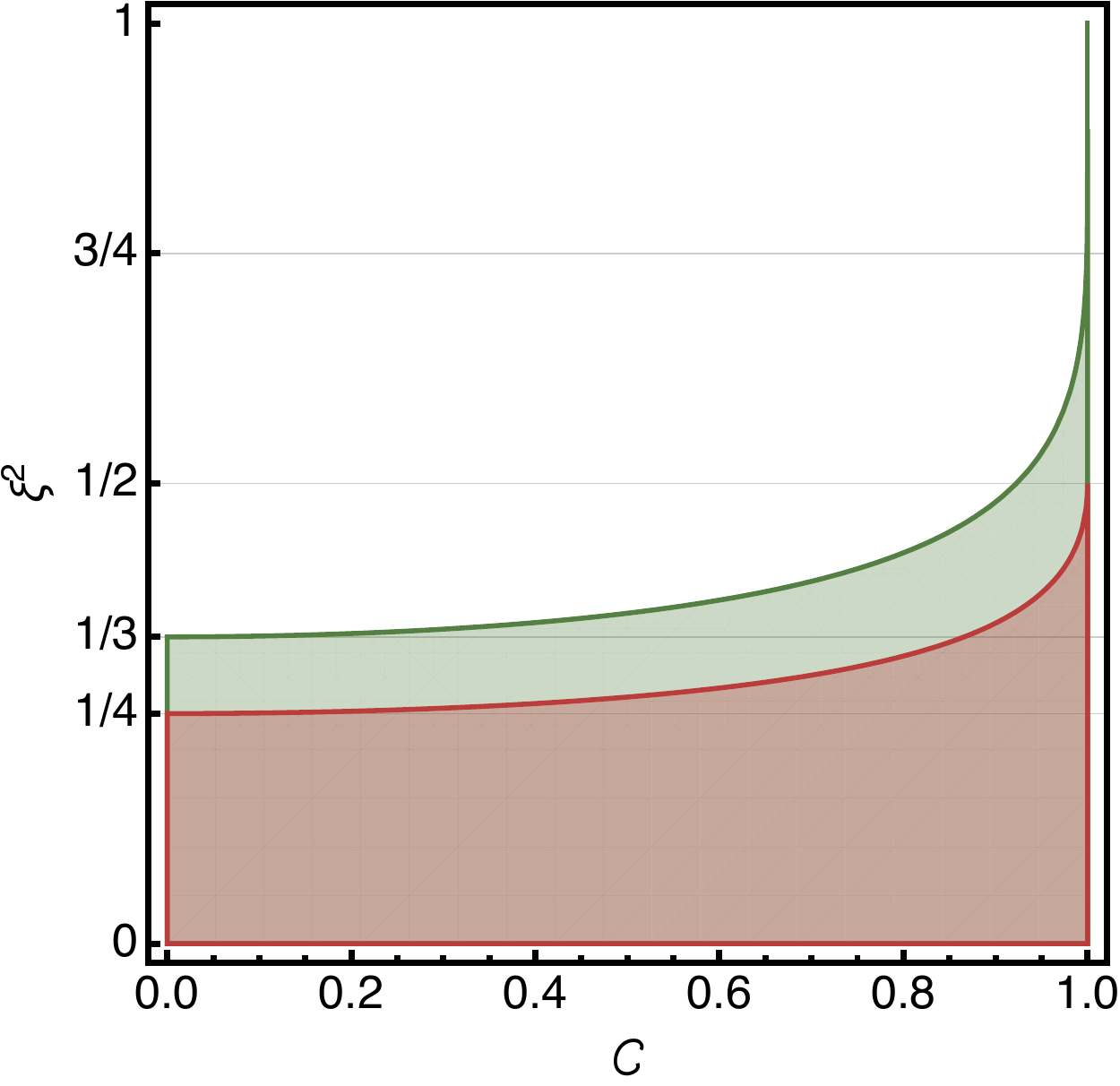}
	\caption{
	\textbf{Regions with Bell correlations.} 
	The figure shows regions in the $\mcC$-$\xi^2$ plane for which Bell correlations are detected by the witnesses given in Eq.~\eqref{eq:1m} (Red) and Eq.~\eqref{eq:2m} (Green).
	Red: For $\xi^2 \leq 1/4$ there are Bell correlations. For $ 1/4 < \xi^2 < 1/2$ the Bell correlation witness Eq.~\ref{eq:1m} is violated only if $\mathcal{C}$ is sufficiently large, while for $\xi^2 \geq 1/2$ there is no violation.
	Green: For $\xi^2 \leq 1/3$ there are Bell correlations. For $1/3 < \xi^2 < 1$ the Bell correlation witness Eq.~\ref{eq:2m} is violated only if $\mathcal{C}$ is sufficiently large.
	}
	\label{fig:1}
\end{figure}

\section{A Bell inequality from a linear QFI bound}
\label{sec:bell-inequality-from-1}

In the previous section, we took a rather free inspiration from the
Heisenberg uncertainty relation to construct a Bell inequality where
only states with a large enough QFI could potentially violate
it. Here, we tackle the problem more directly by linearizing the
right hand side of Eq.~\eqref{eq:12}. We simply use $\langle
B^2 \rangle_{\rho} \leq \lVert B \rVert^2_{\red{\infty}}$, to arrive at
a linear lower bound 
\begin{equation}
\label{eq:2}
\sqrt{\mathcal{F}(\rho,A)} \geq \langle W \rangle_{\rho} \equiv \frac{1}{\lVert B \rVert_{\infty}}\langle i [A,B]  \rangle_{\rho} .
\end{equation}
Here, and in the following, we choose the sign of $B$ such that $\langle W
\rangle_{\rho}$ is positive.

Interestingly, since $\mathcal{F}(\rho,A) \leq N$ for all
separable states, Eq.~\eqref{eq:2} can be turned into
an entanglement witness with operator $\mathcal{W} = \sqrt{N} - W$ for
all $B$. In other words, the correctness of the Heisenberg uncertainty
relation gives us a constructive tool to derive new entanglement witnesses.

This linearization seems to come at the price that the bound is
now much looser, but it turns out that, at least for pure states,
there always exists a $B$ to achieve tightness.
\begin{theorem}
  For $\rho = \left| \psi \right\rangle\!\left\langle \psi \right|
  \equiv \psi$,
  the choice 
\begin{equation}
\label{eq:3}
B = -i [A, \psi]
\end{equation}
implies tightness of Eq.~\eqref{eq:2}.
\end{theorem}
\begin{proof}
  This can be proved by direct calculation. For this, we define the
  orthogonal state
  $| \psi^{\perp} \rangle = 1/(\Delta_{\psi}A) (A - \langle A
  \rangle_{\psi}) \left| \psi \right\rangle $ with
  $\Delta_{\psi}A = \sqrt{\langle (A - \langle A \rangle_{\psi})^2 \rangle_{\psi}}$ and find that $B = i
  \Delta_{\psi} A (\left| \psi^{\perp} \right\rangle\!\left\langle \psi
  \right| - \left| \psi \right\rangle\!\left\langle \psi^{\perp}
  \right| )$ and 
  $\lVert B \rVert_{\infty} = \Delta_{\psi} A$. This leads to
  $i[A,B] = A^2 \psi - 2 A \psi A + \psi A^2$. Hence, one has
  $\langle W \rangle_\psi = 2 V(\psi,A)$. Since for pure states
  $\mathcal{F}(\psi,A) = 4 V(\psi,A)$
  \cite{Braunstein_Statistical_1994}, this implies equality in  Eq.~\eqref{eq:2}.
\end{proof}

We study a specific example for the choice of
Eq.~\eqref{eq:3}. We consider the GHZ state, 
\begin{equation}
  \label{eq:4}
  \left| \mathrm{GHZ} \right\rangle = \frac{1}{\sqrt{2}}(\left| 0 \right\rangle
^{\otimes N} + \left| 1 \right\rangle ^{\otimes N}).
\end{equation}
This state has the maximal QFI for $A = S_z$ with $\mathcal{F}(\mathrm{GHZ},S_z) = N^2$.
Again, direct calculation shows that $W = N (\left| \mathrm{GHZ}
\right\rangle\!\left\langle \mathrm{GHZ} \right| - | \mathrm{GHZ}^{\perp}
\rangle\!\langle \mathrm{GHZ}^{\perp}| )$, where $|
\mathrm{GHZ}^{\perp} \rangle =1/\sqrt{2}(\left| 0 \right\rangle ^{\otimes N }
- \left| 1 \right\rangle ^{\otimes N})$. 

Interestingly, the very same $W$ as in Eq.~\eqref{eq:2} appears when
quantum mechanics is applied to the Bell inequality of Mermin \cite{Mermin_Extreme_1990}, up to a constant and an irrelevant
phase. With our choice of the  normalization, Mermin's inequality show Bell correlations of $\rho=\left| \mathrm{GHZ}
\right\rangle\!\left\langle \mathrm{GHZ} \right| $ whenever
\begin{equation}
\label{eq:5}
\langle W \rangle_{\rho} \leq
\begin{cases}
  N\, 2^{-N/2+1} & \text{$N$ is even},\\
  N\, 2^{-N/2+1/2} & \text{$N$ is odd}.
\end{cases}
\end{equation}
is violated. We compare this to the witness of large QFI
\begin{equation}
\label{eq:6}
\langle W \rangle_{\rho} \leq \sqrt{\mathcal{F}(\rho,S_z)}.
\end{equation}
We observe a connection between the Bell inequality and the lower
bound on the QFI. The GHZ state maximally violates the Bell inequality
and makes Eq.~\eqref{eq:6} being tight. We see that a certain minimal QFI is
necessary to potentially violate the Bell inequality. However, the
fact that $W$ is tailored to the GHZ state makes both inequalities not
very useful for other states. Furthermore, a QFI beyond the shot noise
limit $N$ is not necessary in this case. To illustrate this, we
consider the quantum state 
\begin{equation}
\label{eq:21}
\rho = \frac{1+p}{2} \left| \mathrm{GHZ} \right\rangle\!\left\langle
  \mathrm{GHZ}\right| + \frac{1-p}{2}| \mathrm{GHZ}^{\perp} \rangle\!\langle
  \mathrm{GHZ}^{\perp}|
\end{equation}
with $p \in [0,1]$. Using the PPT criterion, one easily convince
oneself that the state has bipartite entanglement for any $p>0$. It
violates the Mermin inequality if $p > 2^{-N/2+1} N$. Last, the state
has a QFI of $\mathcal{F}(\rho,S_z) = p^2 N^2$, implying that $p >
1/\sqrt{N}$ is necessary to have a QFI that is larger than for any
separable state. For large $N$, the latter bound is exponentially more
restrictive than the local bound of the Mermin inequality.

\section{Discussion}
\label{sec:discussion-outlook}

We investigated whether Bell correlations and large quantum Fisher information (QFI) are connected.
For two examples of Bell inequalities, one instance from a class studied
in \cite{Tura_Detecting_2014} and the Mermin inequality
\cite{Mermin_Extreme_1990}, we showed that a sufficiently large QFI is
necessary for a violation. How generic is this connection? Is it
possible to find a Bell inequality that is violated for any quantum
state with large QFI? Both
approaches presented in this paper give hope to find further Bell
inequalities designed for such states like the Dicke
states. Currently, however, we are not aware of a constructive method
to conjecture and prove these inequalities. This is mainly due to the
step of finding good measurement basis for a quantum operator like in
Eq.~\eqref{eq:14} for a nontrivial Bell inequality.

Finally, note that a large QFI is generally not necessary for a
quantum state to violate a Bell inequality. Indeed, every pure entangled state
can violate a Bell inequality \cite{Gisin,Popescu,Guhne}, but not
every pure entangled state has a QFI beyond the standard quantum limit if only collective measurements are performed \cite{Hyllus_Not_2010}.

\emph{Acknowledgements.---} We would like to acknowledge discussions with Roman Schmied. Financial support by the European ERC-AG MEC and the Swiss national science foundation (Starting grant DIAQ, NCCR-QSIT, and Grant No. $20020\_169591$) is gratefully acknowledged.

\bibliographystyle{apsrev4-1}
\bibliography{ConnectionQFIBell}
\end{document}